\documentclass{article}

\usepackage{todonotes}
\usepackage{biblatex}
\usepackage{amsthm}
\usepackage{comment}
\usepackage{authblk}

\newtheorem{theorem}{Theorem}[section]
\newtheorem{lemma}[theorem]{Lemma}
\newtheorem{corollary}[theorem]{Corollary}
\newtheorem{observation}[theorem]{Observation}
\newtheorem{proposition}[theorem]{Proposition}

\theoremstyle{definition}
\newtheorem{definition}{Definition}[section]

\title{Qubit Routing for (Almost) Free}
\author[1]{Arianne Meijer--van de Griend}
\affil[1]{Department of Computer Science, University of Helsinki, Finland}

\bibliography{generic}

\begin{document}
\maketitle
\begin{abstract}
In this paper, we give a mathematical proof that bounds the number of CNOT gates required to synthesize an $n$ qubit phase polynomial with $g$ terms to be at least $\mathcal{O}(\frac{gn}{\max (\log g, 1)})$ and at most $\mathcal{O}(gn)$. However, when targeting restricted hardware, not all CNOTs are allowed. If we were to use SWAP-based methods to route the qubits on the architecture such that the earlier synthesized gates are natively allowed, we increase the number of CNOTs by a routing overhead factor of $\mathcal{O}(\log n) \leq \alpha \leq \mathcal{O}(n \log^2 n)$.
However, if we only synthesize allowed gates, we do not need to route any qubits. Moreover, in that case the routing overhead factor is $1 \leq \alpha \leq 4 \simeq\mathcal{O}(1)$. Additionally, since phase polynomials and Hadamard gates together form a universal gate set, we get qubit routing for almost free. 
\end{abstract}

\section{Introduction}
The qubit routing problem is a thorn in the side of many quantum platforms. If gates in a circuit do not adhere to a QPUs connectivity constraints, the computational cost of optimally swapping qubits around the architecture is huge~\cite{paler2021nisq,cowtan2019qubit,ito2023algorithmic}, and the resulting circuits are generally much too deep to run on noisy hardware.

However, recent advances in architecture-aware synthesis~\cite{gogioso2022annealing,huang2024redefining,vandaele2022phase,winderl2023recursively} experimentally indicate that it is better to make circuits that fit hardware constraints immediately rather than trying to fix the connectivity constraints after synthesis using SWAP gates. In principle, these algorithms solve the qubit routing problem without routing the qubits. Instead of interpreting the qubit routing problem as a problem of qubit connectivity, architecture-aware synthesis interprets the qubit routing problem as a problem of using the wrong gates. This change in perspective allows architecture-aware synthesis to make smarter choices for gate placement and apply global optimizations. 

Previous research~\cite{jssoftware} has largely focused on experimental results comparing various architecture-aware synthesis approaches with common state-of-the-art SWAP-based transpilers. In this paper, we go a step further and prove mathematically that fixing connectivity constraints after synthesis adds many preventable CNOTs. In contrast, we also prove that the minimum CNOT complexity and the maximum CNOT complexity of architecture-aware synthesis and unconstrained synthesis is the same for phase polynomials. From there we extrapolate that if we write universal circuits as sums over paths, i.e. phase polynomials interleaved with Hadamard gates, we can synthesize any arbitrary circuit for any arbitrary connectivity graph with constant overhead. 

This different approach to the qubit routing problem does not conform to current assumptions in the design of the quantum software stack. Generally, transpilers take as input a well-optimized synthesized circuit and transform them to circuits that fit the target devices native gate set, including connectivity constraints. However, in architecture-aware synthesis it is best to give a circuit comprising of larger logical primitives such that the transpiler does not need to undo the synthesis that originally created the circuit. Note that for any quantum program that is not a small toy example, the programs are defined in these larger logical primitives and then synthesized. 

Thus, the main takeaway from this research is that sparse qubit connectivity is not necessarily as detrimental to circuit execution as SWAP-based transpilers make it seem. We just need to be more mindful about the original circuit generation process and take target hardware earlier into account in the stack.

The structure of this paper is as follows. In Section~\ref{phasepoly}, analyse the CNOT count for phase polynomials. First, we give a short introduction into phase polynomials as they are used as the main primitive in the rest of the paper (Section~\ref{subsec:prelim}). Then, we prove the upper and lower bounds for CNOT count for phase polynomial synthesis (Sections~\ref{subsec:upper} and \ref{subsec:lower}, respectively). For this, we give the bounds for the CNOT count in the unconstrained case, where all CNOTs are allowed, and the constrained case, where CNOTs are only allowed between certain pairs of qubits. Additionally, we use these bounds to calculate the additional CNOTs required when routing a phase polynomial that was not synthesized to fit the target hardware (Section~\ref{subsec:cost}). In Section~\ref{sec:universal}, we generalize the results for arbitrary circuits. Lastly, we end the paper with a discussion on the limitations of this work (Section~\ref{sec:discussion}), future research directions (Section~\ref{sec:future}), and conclusion (Section~\ref{sec:conclusion}).

\section{On the CNOT complexity of phase polynomials}\label{phasepoly}
As a stepping stone to the analysis of the cost qubit routing to universal circuits, we first prove some new bounds for phase polynomial synthesis in both the unconstrained case and the constrained case where CNOTs are only allowed between certain qubit pairs.

\subsection{Phase polynomials}\label{subsec:prelim}
Phase polynomials are circuits consisting of CNOT and $R_Z$ gates. Another way of thinking about them is as a Hamiltonian consisting only of multi-qubit Z-interactions. Each term in the Hamiltonian rotates the qubits with some specified angle based on the qubits that term interacts with. Mathematically, such a term is written as $e^{-i\frac{\alpha}{2}\otimes_{P_i\in \{I, Z\}} P_i}$ where $\alpha$ is the angle of rotation and the qubit interacts with this term if the corresponding $P_i$ is $Z$.

Phase polynomials make a convenient representation because all terms commute with each other. Additionally, we can easily keep track of each term in a binary matrix where each column is a bitstring representing whether that term has a $Z$ in the exponent or not. This matrix is called the parity matrix of a phase polynomial because it has the property that the rotation is applied to the state if and only if the parity of the qubits that interact with the term is $1$. 

Moreover, we can decompose this representation by adding CNOT gates until any of the columns only interacts with a single qubit. At that point, the columns representing a single qubit gate are removed from the column and the corresponding $R_Z(\alpha)$ gate is synthesized (see e.g. \cite{graysynth, steinergray}). When a CNOT gate is synthesized, the parity matrix is updated using binary row addition (modulo $2$). The row corresponding to the target of the CNOT is added to that of the control of the qubit. Essentially, a CNOT toggles the participation of the control qubit depending on the participation of the target qubit in each term. 

Note that the size of the parity matrix is restricted to maximally $2^n-1$ columns and exactly $n$ rows for a phase polynomial with $n$ qubits. This is because the bistrings in the parity matrix are all unique. If they are not, then there are multiple multi-qubit Z-interactions acting on the same qubit. In that case, we can replace them by a single multi-qubit interaction where the angle is sum of the angles of the original interacitons. Additionally, the all-zero bitstring does not exist in the parity matrix because that corresponds to an operation that does not act on any qubits.

\subsection{Upper bounds}\label{subsec:upper}
We start by proving the CNOT upper bounds for phase polynomial synthesis.
\begin{theorem}[Unconstrained phase polynomial synthesis CNOT upper bound]\label{thm:unconstrained-upper}
    A phase polynomial with $g$ terms acting on $n$ qubits can be implemented by a quantum circuit using no more than $\mathcal{O}(gn)$ CNOT gates.
\end{theorem}

\begin{proof}
    We can prove this in two different ways: 
    \begin{itemize}
        \item We can use naive decomposition of each term $p_i$. This will take $2(h(p_i)-1) \leq \mathcal{O}(n)$ CNOTs per term, resulting in $\mathcal{O}(gn)$ CNOTs in total. Where $h(p)$ is the Hamming weight of the parity string describing the term.
        \item Following the logic used in the author's PhD thesis~\cite{thesis}, we need at most 1 CNOT per row and per column in the parity matrix, resulting in $\mathcal{O}(gn)$ CNOTs. However, this leaves the qubits in an entangled state. This can be undone by uncomputing the CNOTs or by synthesizing the linear reversible operation in $\mathcal{O}(\frac{n^2}{\log n})$~\cite{patelmarkovhayes}. Note that synthesis should not take more CNOTs than direct uncomputation of the previously generated CNOTs as this is a circuit that implements the desired action. Thus, phase polynomial synthesis requires at most $\mathcal{O}(gn)$ CNOTs in total.  
        \end{itemize}
\end{proof}

This also holds in the architecture constrained case.

\begin{theorem}[Constrained phase polynomial synthesis CNOT upper bound]\label{thm:constrained-upper}
    A phase polynomial with $g$ terms acting on $n$ qubits can be implemented by a quantum circuit using no more than $\mathcal{O}(gn)$ CNOT gates where each CNOT fits the connectivity constraints specified by some coupling graph $G$.
\end{theorem}

\begin{proof}
    Similar to the unconstrained case, we can prove this in two different ways: 
    \begin{itemize}
        \item We can use naive decomposition of each term. We do this by building a Steiner tree over the interacting qubits in the term to determine the CNOT ladder. The CNOT ladder will have 1 CNOT for each edge in the Steiner tree and 1 CNOT per Steiner node in the Steiner tree. This will take no more than $4h(p_i-1)\leq\mathcal{O}(n)$ CNOTs per term, resulting in $\mathcal{O}(gn)$ CNOTs in total.
        \item Following the logic used in the author's PhD thesis~\cite{thesis}, we need at most 2 CNOTs per row and per column in the parity matrix, resulting in $2gn$ CNOTs. Then, the qubits can be disentangled by uncomputing the CNOTs directly ($+gn$ CNOTs) or they can be synthesized using $\mathcal{O}(\frac{n^2}{\log \delta})$~\cite{SBE}. Thus, synthesis takes at most $\mathcal{O}(gn)$ CNOTs.
    \end{itemize}
\end{proof}

\subsection{Lower bounds}\label{subsec:lower}
In this section, we prove a tighter CNOT lower bound for phase polynomial synthesis than previously known.
We do this by proving that there exists a phase polynomial that needs at least that many CNOTs, similar to Patel et al.~\cite{patelmarkovhayes}. 
This formulation of a lower bound trivially defines lower bounds of $\Omega(0)$ (e.g. only single qubit gates), $\Omega(n)$ (e.g. one term involving all qubits),  $\Omega(\frac{n^2}{\log n})$ (e.g. $n$ linearly independent terms~\cite{patelmarkovhayes}), and $\Omega(2^n)$ (e.g. all possible distinct terms over $n$ qubits~\cite{graysynth}). 
However, we prove a CNOT lower bound that is well behaved for most of these instances.
Additionally, we use $\mathcal{O}$-notation in favor of $\Omega$-notation to account for particular phase polynomial instances that happen to require less CNOTs, such as phase polynomials with only single qubit gates, or phase polynomials that can be written as two parallel phase polynomials on a subset of qubits. Hence, we technically also prove the upper bound of the CNOT lower bound of phase polynomials.

The general outline of the proof is based on the fact that CNOTs entangle and disentangle qubits. In the phase polynomial, this means that CNOT can add or remove a qubit from a term in the phase polynomial. To synthesize a phase polynomial, every qubit needs to be disentangled from every term in the phase polynomial at least once. 

\begin{lemma}\label{lem:disconnect-unconstrained}
    Disconnecting a qubit $q$ from all terms in a phase polynomial with $g$ terms requires at least $\mathcal{O}(\frac{g}{\max(\log g, 1)})$ CNOTs.    
\end{lemma}
\begin{proof}
    We prove this bound constructively, following a similar recursive strategy as is used in the GraySynth algorithm~\cite{graysynth} and a similar divide and conquer strategy used by Patel et al.~\cite{patelmarkovhayes}. First, we can group the term in the phase polynomial into terms that interact with qubit $q$ ($g_1$) or not ($g_0$). The terms that do not interact with this qubit ($g_0$) do not need a CNOT involving the qubit as they are already disconnected.

    The remaining terms ($g_1$) need a CNOT to implement the desired entanglement. So, we choose another qubit ($q_2$) to facilitate this entanglement. Again, some of the terms interact with the second qubit ($g_{11}$) and others do not ($g_{10}$). We can place 1 CNOT between the two chosen qubits to create the entanglement for the $g_{11}$ terms. This disconnects qubit $q$ from the terms in $g_{11}$ but it keeps the terms $g_{10}$ entangled with $q$. Thus, we need to find a third qubit ($q_3$) and do the same for the remaining $g_{10}$ terms, adding 1 more CNOT. We continue this splitting procedure until there are no more terms left.

    Using this procedure, we split the $g$ terms into smaller and smaller sets of terms where each set results in 1 CNOT. We can split the $g$ terms $\mathcal{O}(\log g)$ times, resulting in $\mathcal{O}(\frac{g}{\log g}+1)$ sets. The first set ($g_0$) does not require any CNOTs. Additionally, if $g=1$ then we divide by zero. In that case, qubit $q$ should be connected, and we need one CNOT to disconnect it because if the qubit was not interacting with the only term in the phase polynomial, it was not part of the phase polynomial. 
    Thus, we need at least $\mathcal{O}(\frac{g}{\max (\log g, 1)})$ CNOTs to disconnect one qubit from all terms in the phase polynomial.
\end{proof}

We use this to prove the CNOT lower bound for unconstrained phase polynomial synthesis.

\begin{theorem}[Unconstrained phase polynomial CNOT lower bound]\label{thm:unconstrained-lower}
    There exists an $n$ qubit phase polynomial with $g$ terms needs to be implemented with at least $\mathcal{O}(\frac{gn}{\max (\log g, 1)})$ CNOTs.
\end{theorem}
\begin{proof}
    From Lemma~\ref{lem:disconnect-unconstrained} we know that each qubit needs at least $\mathcal{O}(\frac{g}{\max (\log g, 1)})$ CNOTs to disconnect the qubit from every term.
    This must be done for every qubit in the phase polynomial, resulting in at least $\mathcal{O}(\frac{gn}{\max (\log g, 1)})$ CNOTs.
\end{proof}

Note that this CNOT lower bound corresponds to the CNOT counts for some phase polynomials with known optimal decompositions. Namely,
\begin{itemize}
    \item a phase polynomial with a single term ($g = 1$) involving all qubits: \\
    $\mathcal{O}(\frac{1n}{\max(\log 1,1)}) = \mathcal{O}(n)$,
    \item a phase polynomial with $g = n$ linearly independent terms: \\$\mathcal{O}(\frac{n\cdot n}{\max(\log n,1)}) = \mathcal{O}(\frac{n^2}{\log n})$,
    \item a phase polynomial with all $g = 2^n$ terms: \\$\mathcal{O}(\frac{2^nn}{\max(\log 2^n,1)}) = \mathcal{O}(2^n) = \mathcal{O}(g)$
\end{itemize}

Next, we consider the hardware constrained case of phase polynomial synthesis where CNOTs are only allowed between certain pairs of qubits. 

\begin{observation}\label{obs:constrained>unconstrained}
    The CNOT lower bounds in the constrained case is at least the CNOT lower bounds for the unconstrained case.
    Otherwise, we could trivially constrain the unconstrained case to do better.
\end{observation}

Intuitively, the constrained case should require more CNOTs. Thus, we repeat the analysis for the CNOT lower bound for the constrained case to see how many more CNOTs are needed than in the unconstrained case.

\begin{lemma}\label{lem:disconnect-constrained-most}
    Disconnecting a qubit $q$ from all $g$ terms of a phase polynomial takes \textbf{at most} $\mathcal{O}(\frac{g}{\max(\log g, 1)})$ CNOTs where each CNOT adheres to the connectivity constraints specified by a coupling graph $G$.
\end{lemma}
\begin{proof}
    We can disconnect every term in a phase polynomial under connectivity constraints by combining the divide and conquer strategy from Lemma~\ref{lem:disconnect-unconstrained} with depth-first search.
    
    First, we define $G'$ to be some spanning tree over $G$ with qubit $q$ as root. Similar to Lemma~\ref{lem:disconnect-unconstrained}, we split the terms $g$ into $g_1$ and $g_0$ depending on whether qubit $q$ interacts with the term of not, respectively. Then, we pick some second qubit $q_c$ that is a child of qubit $q$ in $G'$ and split the terms $g_1$ into $g_{11}$ and $g_{10}$ depending on whether the terms interact with $q_c$ or not, respectively. Then, we want to do two things: (1) disconnect qubit $q$ from the terms in $g_{11}$, and (2) SWAP qubits $q$ and $q_c$. This can be done with two CNOTs. We continue to traverse graph $G'$ with $q$, taking only 2 CNOTs and splitting the remaining terms at each step. 
    When we reach a leaf and need to go back, we uncompute the previously placed CNOTs. This works because at the point of uncomputation, we are only considering terms of the phase polynomial that did not disentangle qubit $q$ when those CNOTs were originally placed. 

    Using this strategy, we place at most 4 CNOTs (2 for going down an edge and 2 for going up an edge in $G'$) for every time we split the terms in the phase polynomial. Thus, following the logic in Lemma~\ref{lem:disconnect-unconstrained}, this takes $\mathcal{O}(\frac{g}{\max(\log g, 1)})$ CNOTs.
\end{proof}

\begin{comment}
    TODO: Figure CNOT + SWAP = 2CNOT
\end{comment}

\begin{corollary}\label{cor:disconnect-constrained-least}
    Disconnecting a qubit $q$ from all $g$ terms of a phase polynomial takes \textbf{at least} $\mathcal{O}(\frac{g}{\max(\log g, 1)})$ CNOTs where each CNOT adheres to the connectivity constraints specified by a coupling graph $G$.
\end{corollary}
\begin{proof}
    From Observation~\ref{obs:constrained>unconstrained}, we know that the CNOT lower bound for disconnecting a single qubit from each term in the phase polynomial is at least $\mathcal{O}(\frac{g}{\max(\log g, 1)})$ and from Lemma~\ref{lem:disconnect-constrained-most}, we know that it takes at most $\mathcal{O}(\frac{g}{\max(\log g, 1)})$ CNOTs. Hence, the bound is tight.
\end{proof}

Repeating this strategy gives the CNOT lower bound for constrained synthesis.

\begin{theorem}[Constrained phase polynomial synthesis CNOT lower bound]\label{thm:constrained-lower}
    There exists an $n$ qubit phase polynomial with $g$ terms that must be implemented with at least $\mathcal{O}(\frac{gn}{\max (\log g, 1)}$ CNOTs where each CNOT adheres to the connectivity constraints specified by some coupling graph $G$.
\end{theorem}
\begin{proof}
    Similar to the proof of Theorem~\ref{thm:unconstrained-lower}, we need to disconnect each qubit from all terms in the phase polynomial. 
    So, we apply the result from Corrolary~\ref{cor:disconnect-constrained-least} to all $n$ qubits, resulting in at least $\mathcal{O}(\frac{gn}{\max (\log g, 1)})$ CNOTs.
\end{proof}

Note that if the spanning tree $G'$ is a line, i.e. if the graph $G$ contains a Hamiltonian path, then it takes less CNOTs than when $G'$ is some binary tree. This is because the uncomputation when backtracking in the depth-first traversal doubles the CNOT count. So, there is an effect of the graph structure on the CNOT count, but it is a scalar factor. 

\subsection{The cost of qubit routing phase polynomials}\label{subsec:cost}
Now that we know the CNOT cost of synthesizing phase polynomials, we investigate the SWAP cost routing the qubits of a circuit that does not fit the target hardware. We do this by first defining the \textbf{routing overhead factor} so we have a theoretical measure of cost in terms of scaling in CNOT complexity.

\begin{definition}[Routing overhead factor]
    The \textbf{routing overhead factor} is the scalar \textbf{$\alpha$} such that $\alpha C_{OPT}=C_{Routed}$ where $C_{OPT}$ is the optimal CNOT count and $C_{Routed}$ is the routed CNOT count.
\end{definition}

Then, we can calculate the upper and lower bounds of the routing overhead factor when using SWAP-based routing.

\begin{theorem}\label{thm:swap-upper}
    An $n$ qubit phase polynomial with $g$ terms that has been synthesized without taking connectivity into account and then routed using a practical algorithm requires at most $\mathcal{O}(gn\cdot n\log^2n)$ CNOTs.
\end{theorem}
\begin{proof}
    From Theorem~\ref{thm:unconstrained-upper}, we know that we need to route at most $\mathcal{O}(gn)$ CNOTs. In the worst case, none of the CNOTs can be parallelized and the respective qubits need to be moved far away. Thus, we need a sorting network for each CNOT. Although sorting networks can be constructed in $\mathcal{O}(n\log n)$, this algorithm has a large scalar factor~\cite{aksnetwork}, so in practice an $\mathcal{O}(n\log^2n)$ algorithm is used~\cite{batchersorting,pairwisesorting}.
\end{proof}
\begin{corollary}
    \label{cor:swap-overhead-upper}
    The routing overhead factor for routing a phase polynomial that was synthesized without taking into account the CNOT constraints is at most $\mathcal{O}(n\log^2n)$.
\end{corollary}

However, we might get lucky and the phase polynomial was optimized with the minimal number of CNOTs and all CNOTs can be parallelized as much as possible.
\begin{lemma}\label{lem:min-depth}
    An $n$ qubit circuit containing $\mathcal{O}(C)$ CNOTs has a depth of at least $\mathcal{O}(\frac{C}{n})$.
\end{lemma}
\begin{proof}
    You can do at most $\frac{n}{2}$ CNOTs in parallel.
\end{proof}

\begin{lemma}\label{lem:swap_cost}
    A circuit with a CNOT depth of $l$, requires at least $\mathcal{\Omega}(ln\log n )$ additional SWAPs (and hence CNOTs) to move qubits using traditional SWAP-based routing strategies.
\end{lemma}
\begin{proof}
    Before each layer of CNOTs, we need SWAPs to change the locations of the qubits such that the next layer of SWAPs can be applied. This problem defines a sorting network which theoretically can be implemented with $\mathcal{\Omega}(n\log n)$ SWAPs~\cite{aksnetwork}.
\end{proof}
\begin{theorem}
\label{thm:swap-lower}
    An $n$ qubit phase polynomial with $g$ terms that has been synthesized without taking connectivity into account and then routed requires at least $\mathcal{O}(\frac{gn}{\max(\log g, 1)}\cdot \log n)$ CNOTs.    
\end{theorem}
\begin{proof}
    From Theorem~\ref{thm:unconstrained-lower}, we know that the circuit will have at least $\mathcal{O}(\frac{gn}{\max(\log g, 1)})$ CNOTs. From Lemma~\ref{lem:min-depth}, we know that there are at least $\mathcal{O}(\frac{g}{\max(\log g, 1)})$ CNOT layers. From Lemma~\ref{lem:swap_cost}, we know that such a circuit can be routed using at least $\mathcal{O}(\frac{g\cdot n\log n}{\max(\log g, 1)})$ CNOTs.
\end{proof}
\begin{corollary}
    \label{cor:swap-overhead-lower}
    The routing overhead factor for routing a phase polynomial that was synthesized without taking into account the CNOT constraints is at least $\mathcal{O}(\log n)$.
\end{corollary}

However, if the synthesized phase polynomial was synthesized while adhering to the restrictions of the CNOT placement, routing is not needed.
\begin{corollary}
    \label{cor:swap-overhead-routed}
    The routing overhead factor for routing a phase polynomial that was synthesized taking into account the CNOT constraints is exactly $\Theta(1)$.
\end{corollary}

\begin{theorem}[Qubit routing for free in phase polynomials]
    \label{thm:routing4freepp}
    Any $n$ qubit circuit representing a phase polynomial synthesized without taking connectivity into account has a routing overhead factor of $\mathcal{O}(\log n) \leq \alpha \leq \mathcal{O}(n\log^2n)$ when using conventional SWAP-based routing, and a routing overhead of $\alpha = \mathcal{O}(1)$ when resynthesizing it using constrained synthesis
\end{theorem}
\begin{proof}
    The routing overhead ratio for routing an incompatible phase polynomial is $\mathcal{O}(\log n) \leq \alpha \leq \mathcal{O}(n\log^2n)$ (Corollary~\ref{cor:swap-overhead-lower} and \ref{cor:swap-overhead-upper}). Alternatively, the lower and upper bound for phase polynomial synthesis in the constrained case is the same as in the unconstrained case with only a scalar $\leq  4$ difference (Theorem ~\ref{thm:unconstrained-lower}, \ref{thm:unconstrained-upper}, \ref{thm:constrained-lower}, \ref{thm:constrained-upper}). Moreover, if the phase polynomial is synthesized taking the connectivity constraints into account, it does not require additional SWAP gates for routing (Corollary~\ref{cor:swap-overhead-routed}).
\end{proof}
Thus, solving the qubit routing problem for phase polynomials by routing (i.e. swapping) qubits is inefficient. One can do better by synthesizing the phase polynomial and generating gates that are allowed. Moreover, in that case, the number of CNOTs is only a small constant factor ($\leq 4$) times the original CNOT count in the case the original circuit was well optimized.

\section{Extending phase polynomial bounds to general circuits}\label{sec:universal}
Ofcourse, phase polynomials are not universal for quantum computation. In fact, they are classically simulable~\cite{graysynth}. Thus, we need to extend the class of circuits with some single qubit gates to make it universal.

\begin{comment}
    
\begin{proposition}\label{prop:cnot-dihedral}
    The number of CNOTs needed for a CNOT dihedral circuit is the same as the number of CNOTs needed for the derived phase polynomial that is created from removing the the NOT gates for the CNOT dihedral circuit.
\end{proposition}
\begin{proof}
    We can write the CNOT dihedral circuit as a sequence of phase polynomials intertwined with NOT gates. Then, for each NOT gate, we can push them through the later phase polynomials to the end of the circuit. As a NOT gate is pushed through the phase polynomial, it flips the sign of the phase of each term that is connected to the qubit that the NOT gate acts on. We do this for all NOTs in the circuit and then merge the phase polynomials into a single phase polynomial~\cite{sarah}. Hence, a CNOT dihedral is the same as a phase polynomial with some NOT gates at the end of the circuit. Thus, they have the same CNOT count.
\end{proof}

CNOT dihedrals are not universal either. 
\end{comment}
To make a universal circuit out of phase polynomials, we need the Hadamard gate. Thus, we can write any quantum circuit as a sequence of phase polynomials intertwined with a layer of Hadamard gates on at least one qubit. 

\begin{lemma}\label{lem:universal-upper}
    A CNOT$+H+R_Z$ circuit containing $h$ Hadamard gates and $g$ $R_Z$ gates requires at most $\mathcal{O}(gn)$ CNOTs in both the constrained and unconstrained case.
\end{lemma}
\begin{proof}
    First, we split the circuit in $h+1$ sections by cutting the circuit at each Hadamard gate. The resulting subcircuits are all phase polynomials with $g_i$ terms, where $g=\sum_{i=0}^{h+1}g_i$. 
    By Theorem~\ref{thm:unconstrained-upper} and \ref{thm:constrained-upper}, each phase polynomials requires are most $\mathcal{O}(g_in)$ CNOTs in both the constrained and unconstrained case, resulting in at most $\mathcal{O}(\sum_{i=0}^{h+1}g_i n)=\mathcal{O}(n\sum_{i=0}^{h+1}g_i)=\mathcal{O}(gn)$.
\end{proof}
However, for the lower bound, we need to know how to split the circuit into $\frac{h}{2}+1 \leq h_{\min} \leq h+1$ phase polynomials such that the CNOT count is minimized. For this, the number of Hadamard gates needs to be minimized as well. This can be obtained using the TODD algorithm~\cite{todd,vivientodd,vandaele2025lower}. 

Additionally, although a CNOT$+R_Z$ circuit is a phase polynomial, the number of terms in the phase polynomial is only equal to the number of $R_Z$ gates if each $R_Z$ gate acts on a unique parity. Otherwise, two $R_Z$ gates with angle $\theta_1$ and $\theta_2$ that act on the same parity can be replaced by a single $R_Z$ gate with angle $\theta_1+\theta_2$, this is called Phase folding~\cite{vandaele2024optimal,vandaele2025optimal}.
\begin{lemma}
    \label{lem:universal-lower}
    A CNOT+$H+R_Z$ circuit containing $h$ gates and $g$ $R_Z$ gates requires at least $\mathcal{O}(\texttt{argmin}_{g_0...g_{h_{\min}}}\sum_{i=0}^{h_{\min}}\frac{g_in}{\max(\log g_i, n)})$ CNOTs in both the constrained and unconstrained case where the \texttt{argmin} determines the optimal split of the circuit into subcircuit phase polynomials.
\end{lemma}
\begin{proof}
    Similar to the upper bound, we split the circuit on Hadamard gates. However, it is possible that some Hadamards can be parallelized. Additionally, although the resulting CNOT$+RZ$ circuits are phase polynomials, they might contain terms that commute through the layers of Hadamard gates and could be synthesized as part of a different phase polynomial. As such, we use $\texttt{argmin}_{g_0...g_{h_{\min}}}$ to define some optimal split of the circuit into phase polynomials. Then, by Theorem~\ref{thm:unconstrained-lower} and \ref{thm:unconstrained-lower} each of these phase polynomials require at least $\mathcal{O}(\frac{g_i n}{\max(\log g_i, 1)})$ CNOTs, resulting in $\mathcal{O}(\texttt{argmin}_{g_0...g_{h_{\min}}}\sum_{i=0}^{h_{\min}}\frac{g_in}{\max(\log g_i, n)})$ CNOTs at least.
\end{proof}
Although this Lemma holds in both the constrained and unconstrained case, they might have a different optimal split of the circuit into phase polynomials. However, in the case we use the optimal split of the unconstrained case to synthesize a circuit in the constrained case, this will only differ in CNOT count by a constant factor.

\begin{theorem}[Qubit routing for free]
    A CNOT$+H+R_Z$ circuit defined by a sequence of phase polynomials $g_0...g_{h_{\min} }$ interleaved by $h_{\min} $ (incomplete) layers of Hadamard gates and said phase polynomials have been decomposed with unconstrained synthesis, has a routing overhead factor of $\mathcal{O}(\log n) \leq \alpha \leq \mathcal{O}(n\log^2n)$ when using conventional SWAP-based routing, and a routing overhead of $\alpha = \mathcal{O}(1)$ when resynthesizing it using constrained synthesis. 
\end{theorem}
\begin{proof}
    Using Theorem~\ref{thm:routing4freepp}, each phase polynomial has a routing overhead factor of $\mathcal{O}(\log n) \leq \alpha \leq \mathcal{O}(n\log^2n)$ when using SWAP-based routing and a routing overhead ratio of $\alpha = \mathcal{O}(1)$ when using architecture-aware synthesis. Since these are factors independent of $g_i$, they can be moved out of sum in Lemma~\ref{lem:universal-upper} and \ref{lem:universal-lower}, proving the Theorem.
\end{proof}

Hence, by architecture-awarely synthesizing rather than routing a circuit that was synthesized without taking connectivity into account, we get qubit routing for free.
  
\section{Discussion}\label{sec:discussion}
The results of this paper hold for the specific setting of circuits without the use of additional ancillas or classically controlled gates. If this was allowed, it might be possible to get different trade-offs between synthesis and SWAP-based routing.
Additionally, we did not take into account the effect of qubit mapping on the architecture-aware synthesis. Although the algorithmic bounds hold regardless of qubit placement, it is of course desirable to place qubits that often interact with each other close to each other on the hardware. 

However, the main limitation of this result is that the path sum representation needs to be well optimized. If it is not, the phase polynomials might be unnecessarily large and thus result in suboptimal circuits. For example, there exist large phase polynomials that are non-trivially equivalent to identity called Spider nest identities~\cite{spidernests}. Similarly, there are cases where terms from phase polynomials on different sides of a Hadamard gate commute and cancel other terms out. In the case of the Multi-Pauli Commutation Relation~\cite{MCR}, these can be non-obvious to spot. Thus, the results in this paper are theoretical in nature. There is still a lot of work to be done in optimizing the path sum representation before synthesis to improve results.

Similarly, we can slightly improve gate counts by treating Clifford gates as special members in phase polynomials + Hadamard circuits. We can write the terms in the phase polynomials as Pauli string operators and push all Clifford gates in the circuit to one side (usually the end of the circuit). This results in a circuit consisting of a sequence of commuting Pauli Exponentials and a Clifford Tableau. Doing this removes any Clifford angle terms from the phase polynomials, making them smaller. Since the number of gates that can be synthesized from a Clifford Tableau is bounded~\cite{aaronson2004improved,winderl2023architecture}, these gates will be optimized away in large circuits.

The findings of this paper technically apply directly to Clifford$+T$ circuits, rather than Clifford$+R_Z$. However, this is not recommended due to the large sequences of $SHT$ gates that are generated when synthesizing an $R_Z$ gate~\cite{rossselinger}. These sequences highly overestimate the CNOT requirements of the circuit. Thus, you want to either synthesize the $R_Z$ as late as possible, or temporarily replace the large sequences of single qubit gates with a parameter so that they are treated as a singular term. So, this boils back to the problem of having to start with a well-optimized representation.

That being said, in practical use cases, a circuit that needs to be routed comes from somewhere. Usually, the circuit is defined by some abstract representation from which it is synthesized before routing. The main takeaway of this paper is that during this initial synthesis step, it is best to immediately generate gates that fit the target hardware rather than generating incompatible gates and needing to fix them with SWAP-based routing. This does increase the classical computational needs for synthesis, but to the benefit of improved quantum runtime. As long as classical computational resources are abundant and quantum resources scarce, this is a worthwhile trade-off.

\section{Future work}\label{sec:future}
This paper marks the beginning of a new era in quantum circuit transpilation. If we want to make optimal use of quantum computers, we need to step away from SWAP-based routing and instead start synthesizing gates that fit the target hardware. This is a great departure from the existing structure of the quantum stack that starts with a well-optimized synthesized quantum circuit that then needs to be routed. The best way to do this is to get rid of the circuit as a high-level input format for the stack and instead use a format with more complicated primitives~\cite{heunen2026quantum, qsaner}. 

These more complicated primitives then need architecture-aware synthesis variants so that they can be efficiently executed on target hardware. Some variants of some primitives already exist for nearest-neighbor architectures (e.g. Block-encoding~\cite{foqs-lcu}), but they should be extended to arbitrary graphs since real devices rarely implement a perfect grid (e.g. IBM heavy-hexagon~\cite{heavyhex}). 

The results from this paper are good news for devices with stationary qubits and sparse connectivity. However, it can also be interesting for flying qubits that require shuttling for their multi-qubit operations. If the quantum stack moves towards late synthesis of quantum circuits, then the synthesis could also take shuttling constraints into account such that gates are synthesized to require less shuttling.

Lastly, it is not known how tight the proposed CNOT lower bound for phase polynomial synthesis is. We have shown that it is tight for known optimal cases, but can we reduce the upper bound further? If we can, we might be able to make better algorithms. If we cannot, which phase polynomials require more CNOTs and why?
There is still a lot to figure out.

\section{Conclusion}\label{sec:conclusion}
In conclusion, we proposed new CNOT lower bounds for phase polynomial synthesis in both the connectivity constrained and unconstrained cases. The surprising result is that these lower bounds are the same. Similarly, the upper bounds for phase polynomial synthesis in the connectivity constrained case are the same as the unconstrained case. 
We use this to conclude that we can solve the qubit routing problem for free if we use architecture-aware synthesis rather than SWAP-based routing. 

Importantly, our results show that routing a well-optimized circuit with SWAPs is not sufficient to obtain the smallest circuit possible. We cannot synthesize optimally and then fix connectivity constraints afterwards; this needs to be done jointly.
As a result, we need to rethink the abstraction levels in the quantum software stack.

\section*{Acknowledgements}
This research has been funded by Business Finland Project Enhanced Middleware for Quantum Software (EM4QS) and a Postdoc fellowship from the Finnish Quantum Flagship.
I would like to thank Ville Kortovirta, Frans Perkkola, Leo Becker, and Valter Uotila for useful discussions and encouragements.

\printbibliography

\end{document}